\documentclass[12pt,a4paper]{article}
\usepackage{amssymb}
\usepackage{amscd}
\usepackage{mathrsfs}
\usepackage{eufrak}
\def\principaladviser#1{\gdef\@principaladviser{#1}}
\usepackage[centertags]{amsmath}
\usepackage{amsfonts}
\usepackage{amssymb}
\usepackage[english]{babel}
\usepackage{graphicx}
\usepackage{blindtext}
\usepackage{amsthm}
\usepackage{gnuplottex}
\usepackage{epsfig}

\usepackage{amssymb,latexsym}
\usepackage{amsfonts,amsmath}
\usepackage[T1]{fontenc}
\usepackage{amsmath,amssymb}
 
\usepackage{latexsym}
\usepackage{amssymb}
\usepackage{amsfonts}
\usepackage[arrow,frame,matrix]{xy}
 
\usepackage{fullpage}
\usepackage{hyperref}
\hypersetup{
	colorlinks=true,
	linkcolor=blue,
	filecolor=magenta,
	urlcolor=cyan,
}
\graphicspath{ {./Images/} }

\usepackage[margin=2.6cm]{geometry}
\usepackage{placeins}
\usepackage{grffile}
\usepackage{color}
\usepackage{hyperref}
\hypersetup{colorlinks,linkcolor={blue},citecolor={blue},urlcolor={red}}  
\usepackage{float}
\usepackage{subcaption}
\usepackage{ragged2e}
\usepackage[strings]{underscore}
\usepackage{parskip}

\newtheorem{theorem}{Theorem}[section]

\newtheorem{lemma}[theorem]{Lemma}

\newtheorem{definition}[theorem]{Definition}
\newtheorem{proposition}[theorem]{Proposition}

\def\be{\begin{equation}}
\def\ee{\end{equation}}
\def\bea{\begin{eqnarray}}
\def\eea{\end{eqnarray}}

\begin{document}

\title{Perfect Hermitian rank-metric codes}

\author{Usman Mushrraf}


\maketitle
\begin{abstract}
This study investigates Hermitian rank-metric codes, a special class of rank-metric codes, focusing on perfect codes and on the analysis of their packing properties. Firstly, we establish bounds on the size of spheres in the space of Hermitian matrices and, as a consequence, we show that non-trivial perfect codes do not exist in the Hermitian case. We conclude the paper by examining their {packing density}. 
\end{abstract}


\section{Introduction}\label{Sec: 1}
Rank-metric codes have recently received attention for their applications in random linear network coding and not only \cite{koetter2008coding,silva2008rank,gabidulin2010correcting,bartz2022rank}. These codes can be described as subsets of $\mathbb{F}^{m \times n}_{q}$ equipped with the rank metric. The concept of rank-metric codes was introduced by  Delsarte \cite{delsarte1978bilinear} and independently by Gabidulin \cite{gabidulin1985theory} and Roth \cite{roth1991maximum}. Gabidulin used Frobenius automorphism to construct an optimal family of linear codes.     

In 1975, Delsarte and Goethals started investigating the restricted matrices equipped with the rank distance, in connection with alternating bilinear mappings \cite{delsarte1975alternating}. The theory of restricted matrices was motivated in \cite{delsarte1978bilinear} and \cite{delsarte1975alternating} by applications in classical coding theory. Indeed, the evaluations of forms found in \cite{delsarte1975alternating} led to the discovery of second-order Reed-Muller codes, including identifying Kerdock codes and constructing the Delsarte-Goethals chain of codes. Further, restricted rank-metric codes have been studied in
\cite{zhou2020equivalence,  trombetti2021maximum, bik2023higher, gabidulin2004symmetric, kadir2022encoding, mushrraf2024perfect, schmidt2010symmetric, puchinger2018construction, de2021hermitian, abiad2024eigenvalue, gruica2023rank}.

Perfect codes are those codes that achieve the Hamming/Sphere-packing bound with equality, and they represent the codes with maximum efficiency of their use of the ambient space. Remarkable examples of perfect Hamming-metric code are the one introduced by Hamming \cite{hamming1950error} in 1950 which corrects a single error, and the Golay codes presented by Golay \cite{golay1949notes} in 1949 which corrects multiple errors. Perfect codes in the Hamming metric are rare and highly structured,  \cite{huffman2010fundamentals}. Also, in the context of rank-metric (eventually restricted) codes one can prove Hamming/sphere-packing bound. So we can give the notion of perfect codes also in this context. Loidreau in \cite{loidreau2006properties} studied perfect codes in the context of rank-metric codes and proved that there are no non-trivial perfect codes in this case. Later, \cite{abiad2024eigenvalue} and \cite{mushrraf2024perfect} discussed perfect codes in restricted rank-metric codes.

Perfect codes form a special class of covering codes, { indeed a perfect code $\mathscr{C}$ of length $n$ has covering radius $(d-1)/2$, where $d$ is its minimum distance. This means that for every vector $v\in \mathbb{F}^n_q$, there exists a codeword $c$ such that the distance between $v$ and $c$ is at most $(d-1)/2$. Covering codes have been widely studied in the rank-metric, for further details see \cite{bartoli2024saturating,bonini2023saturating,bonini2024geometry,zullo2024saturating}. This work contributes to this growing field by examining the relation of packing density and covering codes in Hermitian rank-metric spaces}. 
\begin{definition}
   A matrix $A\in \mathbb{F}^{n\times n}_{q^2}$ is said to be Hermitian if it satisfies the condition $A^{*}=A$, where $A^{*}$ denotes the conjugate transpose of $A$. The set of all $n$-order Hermitian matrices over the field $\mathbb{F}_{q^2}$ is represented as $\mathrm{H}_{n}(q^2)$.
\end{definition}
Note that $\mathrm{H}_n(q^2)$ is a vector space over $\mathbb{F}_q$ of dimension $n^2$, and so $|\mathrm{H}_n(q^2)|=q^{n^2}$. Hermitian rank-metric codes are defined as $\mathbb{F}_q$-subspaces in $\mathrm{H}_{n}(q^2)$ endowed with rank distance. The rank distance, defined as the rank of the difference of two matrices $A$ and $B$, is given by 
\[d(A,B)=rk(A-B).\] 
The minimum rank distance of a Hermitian rank-metric code is defined, as usual, as
  \[d(\mathscr{C})=\min\{d(A,B) \colon A,B \in \mathscr{C}, A\ne B\}.\]
 The parameters of a Hermitian rank-metric code $\mathscr{C}\subseteq \mathbb{F}^{n\times n}_{q^2}$ are $(n,|\mathscr{C}|,d(\mathscr{C}))$ and they are related by the Singleton-like bound.
\begin{theorem}[\text{\cite[Theorem 1]{schmidt2018hermitian}}]\label{thm:dim}
    Let $\mathscr{C}$ be a Hermitian rank-metric code in $\mathrm{H}_{n}(q^2)$ of minimum distance $d$. Then 
    \[{dim_{\mathbb{F}_q}(\mathscr{C})\leq n(n-d+1)}.\]   
\end{theorem}
It is stated that $\mathscr{C}$ is a \textbf{Hermitian Maximum Rank Distance} (or \textbf{Hermitian MRD}) code if its parameters satisfy the equality in Theorem \ref{thm:dim}.

By using classical arguments in the Hamming metric, we provide the sphere-packing bound in the context of Hermitian rank-metric codes. We then prove bounds on the size of the balls, which we use to prove that non-trivial perfect codes do not exist in Hermitian rank-metric codes. We describe their packing properties via packing densities, obtaining bounds on packing density.

Section \ref{sec: 2} establishes upper and lower bounds on the size of spheres and balls in the space of Hermitian matrices. In Section \ref{sec; 3}, we prove the sphere packing-bound within the context of Hermitian rank-metric codes and investigate perfect codes. Finally, in Section \ref{sec : 4}, we discuss their packing properties and find bounds on the {packing density}.  
 
\section{Bounds on sphere and ball size}\label{sec: 2}
Given a matrix $M\in \mathrm{H}_{n}(q^2)$, the sphere having radius $r\in \mathbb{N}_0$  and centred at $M$ is defined as \[ S(M,r)=\{ N \in  \mathrm{H}_n(q^2) \colon \mathrm{rk}(M-N)=r\}, \]
and the ball $r\in \mathbb{N}_0$ with center $M$ is defined as
\[B(M,r)=\{ N \in  \mathrm{H}_n(q^2) \colon \mathrm{rk}(M-N)\leq r\},\]
and turns out that
\[ B(M,r)=\cup_{i=0}^r S(M,i). \]
It can be proved that the ball size and the sphere size is independent of its center and we denote the sphere size of radius $r$ by $S_r$ and the ball size of radius $r$ by $B_r$. Also, note that
\begin{equation*}
  B_t=\sum_{r=0}^{t}S_r.  
\end{equation*}
Therefore, $S_r$ represents the number of $n$-by-$n$ Hermitian matrices with rank $r$, and $B_r$ represents the number of $n$-by-$n$ Hermitian matrices and rank less than or equal to $r$. The value of $S_r$ and $B_r$ depends only on $q$, $n$, and $r$, as we will see in a few lines.
The Gaussian binomial coefficient of non-negative integers $m\leq n$ is expressed as
\[bin_{q}(n,m)=\prod_{i=1}^{m}\frac{q^{n-i+1}-1}{q^{i}-1}.\]
We will need the $q^2$-binomial coefficients for which the following hold.
\begin{lemma}\label{bin}
 Let $0\leq m\leq n$ and $q\geq2$.
 Then
 \[bin_{q^2}(n,m)\leq q^{2m(n-m)+2},\]
 and 
 \[bin_{q^2}(n,m)\geq q^{2m(n-m)}. \]
\end{lemma}
\begin{proof}
    The upper bound is given in \text{\cite[Lemma 2.1]{ihringer2015finite}} as  
     \[bin_{q}(n,m)\leq{\frac{111}{32}}q^{m(n-m)},\]
     so that 
     \[ bin_{q^2}(n,m)\leq{\frac{111}{32}}q^{2m(n-m)}\leq q^{2m(n-m)+2}.\]
     It is easy to see that the lower bound of the $q$-binomial coefficient is 
     \[bin_{q}(n,m)\geq q^{m(n-m)},\]
     and so 
     \[bin_{q^2}(n,m)\geq q^{2m(n-m)}.\]     
\end{proof}
 Throughout this paper, we will extensively use the value of $bin_{q^2}(m,n)$  and its bounds. In the following theorem, the value of $S_t$ and $B_t$ has been determined, see \cite{carlitz1955representations} and also \cite{gow2012constant}.
\begin{theorem}[\text{\cite[Theorem 3]{carlitz1955representations}}]\label{lem:rhm}
Let $0 \leq t \leq n$ be an integer. We have 
\[S_t=|\{N\in \mathrm{H}_n(q^2)| rk(N)=t\}|=bin_{q^{2}}(n,t)q^{\frac{t(t-1)}{2}}\prod_{s=1}^{t}(q^s+(-1)^s).\]  
\end{theorem}
Now, we can get the upper and lower bounds of $S_t$ and $B_t$ in the next proposition.
\begin{proposition}\label{prop: Boundsize}
For any $t\in\{0,\ldots,n\}$, we have
\[q^{t(2n-t-1)}\leq S_t\leq q^{t(2n-t+1)+2},\]
and
\[q^{t(2n-t-1)}\leq B_t\leq q^{t(2n-t+1)+3}.\]
\end{proposition}
\begin{proof}
    It is true when $t=0$, so assume that $t\geq1$ and let us start by determining the bounds on $S_t$. By using Lemma \ref{bin}
    \[bin_{q^{2}}(n,t)\leq q^{2t(n-t)+2},\]
    and
    \[\prod_{s=1}^{t}(q^{s}+(-1)^{s})\leq\prod_{s=1}^{t}q^{(s+1)}=q^{t+\frac{t(t+1)}{2}}.\]
    So, we derive
    \begin{equation}\label{eq:1}
       S_t\leq q^{t(2n-t+1)+2}.
    \end{equation}
    For the lower bound on $S_t$, by Lemma \ref{bin}
    \[bin_{q^2}(n,t)\geq q^{2t(n-t)},\]
    and we observe that 
    \[\prod_{s=1}^{t}(q^{s}+(-1)^{s})\geq\prod_{i=1}^{t}q^{(s-1)}=q^{\frac{t(t+1)}{2}-t}.\]
    So,
    \begin{equation}\label{eq:2}
    S_t\geq q^{t(2n-t-1)}. 
    \end{equation}
      From \eqref{eq:1} and \eqref{eq:2}, we get the desired bounds on $S_t$
     \[q^{t(2n-t-1)}\leq S_t\leq q^{t(2n-t+1)+2}.\]
     Now we analyze the case of the ball. We know that a lower bound on the sphere size implies a lower bound on the ball size having the same radius. So, we only need to prove the upper bound on the ball size. We have
 \[ B_t\leq \sum _{i=0}^{t}q^{i(2n-i+1)+2} =q^{t(2n-t+1)+2}\left( 1+\sum _{i=0}^{t-1}q^{2n(i-t)+t^2-i^2-(t-i)} \right).\]
 As, for any $i\leq t-1$  \[ 2n(i-t)+t^2-i^2-(t-i)\leq 2(t-i)(t-n),\]
 so, 
 \[B_t\leq q^{t(2n-t+1)+2}\left( 1+\sum _{i=0}^{t-1} q^{2(t-i)(t-n)}\right),\]
and by changing the variables, we obtain that
        \[ B_t\leq q^{t(2n-t+1)+2}(1+\sum _{j=1}^{t}q^{2j(t-n)}).\]
     Observe that
     \[1+\sum _{j=1}^{t}q^{2j(t-n)}=\frac{1-q^{2(t+1)(t-n)}}{1-q^{2(t-n)}}\leq \frac{1}{1-q^{2(t-n)}}\leq q.\]
     This provides the required result for $B_t$. 
\end{proof} 
In the next section, we will use these bounds to study perfect codes in the context of Hermitian rank-metric codes.
\section{Perfect codes }\label{sec; 3}
In this section, we will investigate the sphere-packing bounds, and show that there are no non-trivial perfect codes.
\begin{theorem}[Sphere-packing bound]
Let $\mathscr{C}\subseteq \mathrm{H}_n(q^2)$ be an Hermitian rank-metric code with minimum distance $d$ and $|\mathscr{C}|=M$. If $t=\lfloor\frac{d-1}{2}\rfloor$, then 
\[MB_{t}\leq q^{n^2}.\]
\end{theorem}
\begin{proof}
Let $c_1,c_2\ldots,c_M $ be the codewords of $\mathscr{C}$. From the triangle inequality of the rank metric, we can see that the intersection of $B(c_i,t)$ and $B(c_j,t)$ is empty for any $i,j\in\{1,2\ldots,M\}$ with $i\neq j$.
So,
    \[\left|\bigcup^{M}_{i=1}B(c_{i,t})\right|=\bigcup^{M}_{i=1}\left| B(c_{i,t})\right|\leq |\mathrm{H}_n(q^2)|=q^{n^2},\]
    which implies that 
    \[\sum^{M}_{i=1}B_{t}\leq q^{n^2},\]
    and hence the desired result is obtained.     
\end{proof}
\begin{definition}
    Let $\mathscr{C}\subseteq \mathrm{H}_n(q^2)$  be a Hermitian rank-metric code. We define $\mathscr{C}$ to be a perfect code if it meets sphere packing-bound with equality.
    \end{definition}
    We start to show that non-trivial perfect codes do not exist when $t=\lfloor\frac{d-1}{2}\rfloor=1$. 
  \begin{lemma}\label{lem:t=1}
      Let $\mathscr{C}$ be a Hermitian rank-metric code in $\mathrm{H}_{n}(q^2)$ with minimum distance $d$ and let $t=\lfloor\frac{d-1}{2}\rfloor=1$, then $\mathscr{C}$ is not perfect.
  \end{lemma}
\begin{proof}
Suppose on the contrary that $\mathscr{C}$ is a perfect code.
\begin{equation}\label{eq:t=1}
  MB_{1}=q^{n^2}. 
\end{equation}
Since $t=1$, we have $d=3$ or $d=4$ and by {Singleton-like bound}, 
\[M\leq q^{n(n-2)}.\]
Using that $B_1=1+\frac{q^{2n}-1}{q+1}$, Equation \eqref{eq:t=1} becomes 
\[q^{2n}-1\leq \frac{q^{2n}-1}{q+1},\]
which implies $q\leq-1$, a contradiction. Therefore $\mathscr{C}$ cannot be perfect.  
  
\end{proof}
 
In the next Lemma, we analyze the case in which $\lfloor \frac{d-1}{2}\rfloor=2$.
 \begin{lemma}\label{lem:t=2}
    Let $\mathscr{C}$ be a Hermitian rank-metric code in $\mathrm{H}_n(q^2)$ with minimum distance $d$ and let $t=\lfloor\frac{d-1}{2}\rfloor=2$ then $\mathscr{C}$ is not perfect.
\end{lemma}
\begin{proof}
    Suppose on the contrary that $\mathscr{C}$ is perfect. Then 
    \begin{equation}\label{eq:t=2}
        MB_{2}=q^{n^2}.
    \end{equation}
    Using that $M\leq q^{n(n-4)}$, from Theorem \ref{lem:rhm}, and that $B_2=1+\frac{q^{2n}-1}{q+1}+\frac{(q^{2n}-1)(q^{2(n-1)}-1)q}{(q+1)(q^2-1)}$, we have \eqref{eq:t=2} that 
    \[q^{n^2-4n}\left(1+\frac{q^{2n}-1}{q+1}+\frac{q(q^{2n}-1)(q^{2n-2}-1)}{(q+1)(q^2-1)}\right)\geq q^{n^2},\]
    that is
    \[\frac{q^{2n}-1}{q+1}+\frac{q(q^{2n}-1)(q^{2n-2}-1)}{q+1)(q^2-1)}\geq q^{4n}-1,\]
  from which we derive
    \[q^{2n-1}+q^2-q-1\geq (q^{2n}+1)(q^2-1)(q+1).\]
      Therefore
    \[q^{2n-1}\geq  q^{2n+3}-q^{2n+1}+q^{2n+2}-q^{2n}+q^3,\] 
    which is a contradiction. 
 
\end{proof}
In the following theorem, we prove perfect codes do not exist in this context.
\begin{theorem}\label{main; thm}
There are no non-trivial perfect Hermitian rank-metric codes.
\end{theorem}
\begin{proof}
   Let $\mathscr{C}$ be a Hermitian rank-metric code with minimum distance $d$ then from Theorem \ref{thm:dim}
    \[{M\leq q^{n(n-d+1)}}.\] Suppose that $\mathscr{C}$ is perfect we have \[q^{n(n-d+1)}B_t\geq MB_t= q^{n^2},\] with $t=\lfloor \frac{d-1}{2}\rfloor$.
        Since, by Proposition \ref{prop: Boundsize} $B_t\leq q^{t(2n-t+1)+3}$, then 
        \begin{equation}\label{Thm:eq&}
           q^{n(n-d+1)+t(2n-t+1)+3}\geq q^{n^2},  
        \end{equation}
        that is 
if $d=2t+1$, this implies that
\[-t^2+t+3\geq 0,\]
which hold for $t\in\{0,1,2\}$.
If $d=2t+2$  then $\eqref{Thm:eq&}$ implies that 
\[-t^2+t+3\geq n.\]
Since $n\geq 2$, so we have that $t\in\{0,1\}$. If $t=0$, then $d\in\{0,1\}$ and $\mathscr{C}$ is perfect if and only if $\mathscr{C}= \mathrm{H}_n(q^2)$, which is a trivial case. The remaining part follows by Lemmas \ref{lem:t=1} and \ref{lem:t=2}.
\end{proof}
 
\section{{Packing density}}\label{sec : 4}

In the previous section, we proved that non-trivial perfect codes do not exist in the context of Hermitian rank-metric codes. {Although the non-existence of perfect codes raises the question of how efficiently we can cover the ambient space. This lead us to study the covering properties of Hermitian rank-metric codes, which  ensure that the ambient space is fully covered by the balls of a fixed radius and centered in codewords. To measure this, we consider the concept of packing density}, which measures how much of the ambient space is covered by the spheres centered in the codewords and of radius $\lfloor \frac{ d(\mathscr{C})-1}{2}\rfloor$.
\begin{definition}
    Let $\mathscr{C}$ be a Hermitian rank-metric code in $\mathrm{H}_n(q^2)$ of size $M$. The \textbf{{packing density}} of $\mathscr{C}$ is defined as 
\[\mathscr{D(C)}=\frac{MB_{t}}{q^{n^2}},\] where $t=\lfloor \frac{d(\mathscr{C})-1}{2}\rfloor$.
\end{definition}
By using the bound of Proposition \ref{prop: Boundsize}, we have the following bounds for the {packing density}.
\begin{proposition}\label{prop:4.2}
Let $\mathscr{C}\subseteq\mathrm{H}_n(q^2)$ be a Hermitian MRD code with minimum distance $d$ and let $t=\lfloor \frac{d-1}{2}\rfloor$. 
The {packing density} $\mathscr{D(C)}$ of $\mathscr{C}$ is upper bounded as follows
\[\mathscr{D(C)}\leq\begin{cases}
q^{-t^2+t+3} &\text{if  }  d=2t+1,\\
q^{-n-t^2+t+3} &\text{if  }  d=2t+2,  
\end{cases} \]
and lower bounded as follows
\[\mathscr{D(C)}\geq\begin{cases}
   q^{-t^2-t} &\text{if  }  d=2t+1,\\
   q^{-n-t^2-t} &\text{if  }  d=2t+2.
\end{cases}\]
\end{proposition}
\begin{proof}
   { Let us consider the case $d=2t+1$. From Theorem \ref{thm:dim} $|\mathscr{C}|= q^{n(n-d+1)}$. By Proposition \ref{prop: Boundsize}, $B_t\leq q^{t(2n-t+1)+3}$, and so \[
\mathscr{D(C)} \leq \frac{q^{n (n - 2t)} \cdot q^{t(2n - t + 1) + 3}}{q^{n^2}} = q^{-t^2 + t + 3},
\]
and by Proposition \ref{prop: Boundsize}, $B_t \geq q^{t(2n-t-1)}$, we have \[
\mathscr{D(C)} \leq \frac{q^{n (n - 2t)} \cdot q^{t(2n - t -1)}}{q^{n^2}} = q^{-t^2 -t}.
\]
In a similar way, we can derive the bounds for $d=2t+2$.}
\end{proof}
It is important to note that when the minimum distance is odd, then the upper and lower bounds of {packing density} depend entirely on the error correction capability of the code. If $\mathscr{C}$ is a Hermitian rank-metric code in $\mathrm{H}_{n}(q^2)$ then we can only get on upper bound of {packing density}.
 
\begin{proposition}\label{prop:cvd}
    Let $\mathscr{C}\subseteq \mathrm{H}_{n}(q^2)$ be a Hermitian rank-metric code with minimum distance $d$ and let $t=\lfloor \frac{d-1}{2}\rfloor$. An upper bound of {packing density} $\mathscr{D}(\mathscr{C})$ of $\mathscr{C}$ for $d\in\{2,3,4,5,6\}$ is
    \[\mathscr{D(C)}\leq \begin{cases}
        \frac{1}{q^n} &\text{ if } d=2,\\
        \\
        \frac{q^{-2n+1}+1}{q+1}&\text{ if } d=3,\\
        \\
       \frac{q^{-3n+1}+q^{-n}}{q+1}  &\text{ if } d=4,\\
       \\
        q^{-4n}\left(1+\frac{q^{2n}-1}{q+1}+\frac{q(q^{2n}-1)(q^{2n-2}-1)}{(q+1)(q^2-1)}\right) &\text{ if } d=5,\\
        \\
         q^{-5n} \left(1+\frac{q^{2n}-1}{q+1}+\frac{q(q^{2n}-1)(q^{2n-2}-1)}{(q+1)(q^2-1)}\right) &\text{ if } d=6.
    \end{cases}\]
\end{proposition}
 We note that when $n$ approaches to infinity, the {packing density} goes to zero when  $d$ is even. While, if we consider a family $\{\mathscr{C}_i\}_i\subseteq\mathrm{H}_{n_i}(q^2)$ of Hermitian MRD codes with minimum distance $d=3$, then the sequence of the packing densities goes to $\frac{1}{q+1}$.
 
 \begin{proposition}
     Let $\{\mathscr{C}_i\}_i$ be a family of Hermitian MRD code in $\mathrm{H}_{n_i}(q^2)$, and $\mathscr{C}_i$ has parameters $(n_i,|\mathscr{C}_i|,d)$. Then
     \begin{itemize}
         \item if $d$ is even, then $\lim_{i\to\infty} \mathscr{D}(\mathscr{C}_i)=0$,
         \item if $d$ is odd and $\lim_{i\to\infty} \mathscr{D}(\mathscr{C}_i)$ exists, then 
         \[ q^{-t^2-t}\leq\lim_{i\to\infty} \mathscr{D}(\mathscr{C}_i)\leq q^{-t^2+t+3},\]
         where $t=\lfloor\frac{d-1}{2}\rfloor$.
     \end{itemize}
 \end{proposition}
 \begin{proof}
     It is easy to see that when minimum distance is even in Proposition \ref{prop:4.2} and Proposition \ref{prop:cvd}, as $i$ tends to infinity the {packing density} $\mathscr{D}(\mathscr{C}_i)$ approaches zero. For the case, if $d$ is odd and $\lim_{i\to\infty} \mathscr{D}(\mathscr{C}_i)$ exist, Proposition \ref{prop:4.2} and Proposition \ref{prop:cvd} clearly show that the {packing density} is independent on the order of the matrices of the code. Therefore $\lim_{i\to\infty} \mathscr{D}(\mathscr{C}_i)$ is bounded as in the statement.
 \end{proof}
 \section{Conclusion}\label{sec : 5}
This proposed work aims to study perfect codes in the context of Hermitian rank-metric codes. This is similar to what happens for the classical rank-metric codes \cite{loidreau2006properties} and for the case of alternating matrices for large values of the field size \cite{abiad2024eigenvalue}. Following the approaches in \cite{loidreau2006properties} and \cite{mushrraf2024perfect}, we proved in Theorem \ref{main; thm} that no non-trivial perfect codes exist in the Hermitian rank-metric code case. We also investigated their packing properties and found that spheres centered at the codewords covered $\frac{1}{q+1}$ of the ambient space, with a minimum distance of three.
\bibliographystyle{abbrv}

\section*{Acknowledgements}
I would like to express my deepest gratitude to Ferdinando Zullo for his invaluable, guidance and support throughout the development of this paper.

Usman Mushrraf \\
Dipartimento di Matematica e Fisica,\\
Universit\`a degli Studi della Campania ``Luigi Vanvitelli'',\\
I--\,81100 Caserta, Italy\\
{{\em usman.mushrraf@unicampania.it}}
\end{document}